\documentclass[conference]{IEEEtran}
\IEEEoverridecommandlockouts
\usepackage{cite}
\usepackage{amsmath,amssymb,amsfonts}
\usepackage{algorithmic}
\usepackage{algorithm}
\usepackage{textcomp}
\usepackage{xcolor}
\usepackage{url}
\usepackage{times}
\usepackage{bm}
\usepackage{theorem}
\usepackage{latexsym}
\usepackage{graphics}
\usepackage{graphicx}

\makeatletter
\def\eqnarray{\stepcounter{equation}\let\@currentlabel=\theequation
\global\@eqnswtrue
\global\@eqcnt\z@\tabskip\@centering\let\\=\@eqncr
$$\halign to \displaywidth\bgroup\@eqnsel\hskip\@centering
  $\displaystyle\tabskip\z@{##}$&\global\@eqcnt\@ne 
  \hfil$\;{##}\;$\hfil
  &\global\@eqcnt\tw@ $\displaystyle\tabskip\z@{##}$\hfil 
   \tabskip\@centering&\llap{##}\tabskip\z@\cr}
\makeatother
\theorembodyfont{\itshape}
\newtheorem{thm}{Theorem}
\newtheorem{lem}{Lemma}

\newtheorem{prob}{Problem}

\theorembodyfont{\rmfamily}
{\bf}{\rm}
{\bf}{\rm}
{\bf}{\rm}

{\itshape}{\rmfamily}

\newenvironment{proof}{\noindent{\it Proof.~~}}{\qed\medskip}

\def\vc#1{\mbox{\boldmath $#1$}}

\newcommand{\qed}{\hspace*{\fill}$\Box$}

\newcommand{\rmd}{\mathrm{d}}
\newcommand{\rmN}{\mathrm{N}}
\newcommand{\rmL}{\mathrm{L}}

\newcommand{\calL}{\mathcal{L}}

\newcommand{\calN}{\mathcal{N}}

\newcommand{\calV}{\mathcal{V}}

\newcommand{\calG}{\mathcal{G}}
\newcommand{\calU}{\mathcal{U}}

\newcommand{\calP}{\mathcal{P}}

\newcommand{\bbR}{\mathbb{R}}

\newcommand{\bbZ}{\mathbb{Z}}
\newcommand{\bbE}{\mathbb{E}}

\newcommand{\PP}{\mathsf{P}}

\newcommand{\SINR}{\mathsf{SINR}}

\newcommand{\LoS}{\mathrm{LoS}}
\newcommand{\NLoS}{\mathrm{NLoS}}

\newcommand{\GP}{\mathrm{GP}}

\newcommand{\dd}[1]{\if#11 1\!\!1 
\else {\if#1C I\!\!\!C
\else {\if#1G I\!\!\!G 
\else {\if#1J J\!\!\!J 
\else {\if#1S S\!\!\!S
\else {\if#1Z Z\!\!\!Z
\else {\if#1Q O\!\!\!\!Q
\else I\!\!#1
\fi} 
\fi}
\fi}
\fi} 
\fi} 
\fi} 
\fi} 

\DeclareSymbolFont{bbold}{U}{bbold}{m}{n}
\DeclareSymbolFontAlphabet{\mathbbold}{bbold}
\newcommand{\onev}{\mathbbold{1}}

\usepackage{fancyhdr}
\usepackage{lipsum}

\fancypagestyle{plain}{
  \fancyhf{}
  \fancyhead[C]{Conference on \LaTeX}     
  \fancyfoot[L]{This is a notice}

}
\usepackage{tikz}
\usetikzlibrary{calc}

\begin{document}
%
%

%

\title{Distributed Collaborative 3D-Deployment of UAV Base Stations
for On-Demand Coverage}

\author{\IEEEauthorblockN{Tatsuaki Kimura\IEEEauthorrefmark{1}
and 
Masaki Ogura\IEEEauthorrefmark{2}}
\IEEEauthorblockA{\IEEEauthorrefmark{1}Department of Information and Communications Technology,
Osaka University, Osaka, Japan, }
\IEEEauthorblockA{\IEEEauthorrefmark{2}Department of Bioinformatic Engineering, 
Osaka University, Osaka, Japan,\\
Emails:  kimura@comm.eng.osaka-u.ac.jp,  m-ogura@ist.osaka-u.ac.jp}}



\maketitle
\begin{tikzpicture}[remember picture, overlay]
\node[text width=22.5cm, align=center] at ($(current page.north) + (0.0,-0.2in)$) 
{\scriptsize \copyright 2020 IEEE. Personal use of this material is permitted. Permission from IEEE must be obtained for all other uses, in any current or future media, including reprinting/republishing this material for advertising or promotional purposes, creating new collective works, for resale or redistribution to servers or lists, or reuse of any copyrighted component of this work in other works.};
\end{tikzpicture}


\begin{abstract}
Deployment of unmanned aerial vehicles (UAVs) performing
as flying aerial base stations (BSs) has a great potential of
adaptively serving ground users 
during temporary events, such as major disasters and massive events. 
%
However, planning an efficient, dynamic, and 3D deployment of UAVs 
in adaptation to dynamically and spatially varying 
ground users is a highly complicated problem 
due to the complexity in air-to-ground channels and interference among UAVs. 
%
%
%
In this paper, we propose a novel distributed 3D deployment method for UAV-BSs
in a downlink network for on-demand coverage. 
Our method consists mainly of the following two parts: 
\textit{\textbf{sensing-aided crowd density estimation}}
and \textit{\textbf{distributed push-sum algorithm}}. 
The first part estimates the ground user density
from its observation through on-ground sensors,
thereby allowing us to avoid the computationally 
intensive process of obtaining 
the positions of all the ground users. 
On the basis of the estimated user density, 
in the second part, 
each UAV dynamically updates its 3D position 
in collaboration with its neighboring UAVs for maximizing the total coverage. 
We prove the convergence of our distributed algorithm
by employing a distributed push-sum algorithm framework. 
%
Simulation results demonstrate that our method can 
improve the overall coverage with a limited number of ground sensors. 
We also demonstrate that 
our method can be applied to a dynamic network in which the density of ground users varies temporally.


\end{abstract}


%
\IEEEpeerreviewmaketitle

\section{Introduction}
Unmanned aerial vehicle (UAV)-enabled networks have been 
gaining substantial attention because of their wide variety of
applications including surveillance, military, and rescue operations~\cite{Vala14}. 
In particular, deployment of UAVs performing as 
flying aerial base stations (BSs) to support existing cellular networks
is a key application~\cite{Moza19}. 
%
Typically, UAV-BSs can establish
a line-of-sight (LoS) connection to on-ground users
with high probability owing to their high altitudes. 
Thus, the user coverage can
be improved significantly in an efficient manner~\cite{Bor16b}. 
%
In addition, since UAVs have high autonomous mobility, 
UAV-BSs can provide connections to on-ground users 
in disaster areas (e.g., flood- and earthquake-affected areas), or 
rural areas more robustly and cost-effectively than 
ground BSs of cellular networks~\cite{Bor16}. 
Furthermore, rapid and flexible deployment of UAV-BSs
can enable them to respond to the occurrence of hotspots 
in sports events and open-air concerts
and can help achieve on-demand coverage for these.  
%

Despite the potential benefits of UAV-BSs, 
determining the most effective deployment of UAVs poses
several research challenges~\cite{Moza19}. 
First, owing to the flexible mobility of UAVs, 
the 3D deployment problem of UAVs exhibits a higher degree of freedom and is 
more complicated than that of ground BSs. 
%
Furthermore, air-to-ground (A2G) channels have
different characteristics 
from a terrestrial one because they are highly 
dependent on the altitudes of the UAVs and 
the blockage effect of obstacles (e.g., buildings)~\cite{Hour14, Hour14b}. 
%
Moreover, the deployment of multiple UAV-BSs induces
an inter-cell interference problem, which may degrade  
the communication quality of users. Since the interference 
received from each UAV  depends additionally on the A2G channel condition, 
the overall characteristics of interference and 
the signal-to-interference-plus-noise-ratio (SINR)
of users are also exceedingly complicated. 
As a result, the optimal 3D deployment of UAVs considering
the above factors is a significantly challenging problem. 

%

Owing to the importance of the UAV deployment problem, 
there has been an increasing number of studies that address
the aforementioned challenges~\cite{Moza16, Kala16, Moza19}. 
However, there are few studies that consider the 
spatial and temporal variations of ground users. 
For example, the works~\cite{Hour14, Moza16}
proposed optimal UAV deployment methods that maximize
the coverage area. 
However, in general, the density of users is 
spatially varied for several reasons, such as 
the higher densities at 
stations and sports events. 
Therefore, to maximize the user coverage, 
it is insufficient to maximize the coverage area, 
and the spatial inhomogeneous  density of users must be considered.

Furthermore, the density of ground users varies temporally 
because they may leave and join the network
at any time and move dynamically at each moment. 
Although several studies~\cite{Kala16, Bor16} proposed 
optimal UAV deployment methods
assuming that  the specific positions of all the users are known, 
these positions may change temporally, and 
tracking all the user movements is unrealistic. 
In addition, these methods were fundamentally 
centralized approaches, which are unsuitable 
for a more dynamic network. This is because 
the computational expense is likely to escalate 
with an increase in the number of users or in the total area. 
Therefore, it is crucial to develop 
an efficient 3D deployment method for UAV-BSs 
that can respond to the spatial and temporal dynamics
of users.

In this paper, we propose a novel distributed 3D deployment
method for UAV-BSs in a downlink network. 
In our method, each UAV dynamically updates its 3D position 
by collaborating with neighboring UAVs
so that the overall coverage of users is maximized in an on-demand manner. 
The distributed nature and incremental updates
of our method enable it to be applied to 
spatially and temporally varying networks.

Our proposed method consists mainly of two key parts
that address the above challenges in spatially and temporally varying 
ground users:
i) {\it sensing-aided crowd density estimation}; and 
ii) {\it distributed push-sum algorithm}. 
The first part is used to estimate the density of 
users based on the information from {\it ground sensors}
deployed in the network. 
As it is computationally intensive 
to obtain all the positions of ground users at each moment in time, 
we assume that ground users are distributed by 
a spatial point process. 
%
However, its intensity (density) function is also generally unavailable. 
For this problem, we assume that the ground sensors can detect users 
nearby them, and estimate the number of users, e.g., by
a video surveillance system with human detection/tracking
methods~\cite{Sale15} and Wi-Fi access points with 
received signal strength indicator~(RSSI)- or channel state information
(CSI)-based sensing methods~\cite{Scha14, Xi14}.
Since the intensity of users is typically considered to be
spatially correlated (e.g., road system), 
we infer the entire intensity function from periodically gathered
sensing results. 
%

In the second step, the UAVs dynamically optimize
their positions in a distributed manner. 
%
Fundamentally, the challenge posed by {\it distributed} 3D UAV deployment
is additional to the present problems
such as the A2G channel characteristics. 
This is because the SINR of a user 
served by a UAV is affected by interference from other UAVs. 
Thus, to control the SINR of the user, 
each UAV needs to consider the positions of all the other UAVs
and update its position so that the SINR of all users 
are improved, and not only its serving users is imporved. 
To address this problem, 
we apply a distributed push-sum algorithm framework~\cite{Kemp03, Tata17}. 
With this framework, UAVs incrementally optimize their 3D positions
by collaborating with their neighbor UAVs. We also prove 
that the algorithm converges to a consensus among the UAVs. 
To the best of our knowledge, this is the first study to develop
a distributed 3D UAV deployment method with guaranteed convergence.  
%
%
%
%
%
Regarding the performance metric of a UAV deployment, 
we adopt the {\it total coverage} of users, i.e., 
the total expected number of users whose SINR exceeds a certain threshold. 
By expressing the coverage probability 
of a user theoretically and analyzing the total coverage, 
we apply this performance metric to
the distributed push-sum algorithm framework. 
Furthermore, we evaluate our method with extensive simulations. 
The results reveal that our method can successfully 
improve the overall coverage with a limited number of ground sensors. 
We also apply our method to a dynamic network scenario
and demonstrate that it can respond to moving hotspots 
and provide coverage to users in an on-demand manner.

The remaining part of this paper is organized as follows. 
Section~\ref{sec-related} summarizes related studies. 
In Section~\ref{sec-model}, we describe our network model in detail. 
Section~\ref{sec-main} presents the proposed method. 
In Section~\ref{sec-simu}, we provide several simulation results. 
Finally, Section~\ref{sec-conc} concludes our paper.

\section{Related Work}\label{sec-related}
%
Since UAV-BSs can be potentially used for
demanding future networks, 
several deployment methods have been proposed.
In particular, several studies considered 
an optimal 3D-deployment of a single UAV. 
Al-Hourani et al.~\cite{Hour14} presented
the optimal altitude of a UAV that maximizes the coverage area. 
They also developed a closed-form expression of 
the LoS probability of a UAV based on a probabilistic
LoS model provided by the International Telecommunication Union 
(ITU-R)~\cite{ITUR03}. 
Bor-Yaliniz et al.~\cite{Bor16} proposed an efficient deployment
method that maximizes the number of covered users. 
%
Alzenad et al.~\cite{Alze18} further 
considered a scenario where users have 
heterogeneous quality of service (QoS) requirements. 
%

In contrast to the single-UAV cases, multiple-UAV
scenarios lead to an inter-cell interference problem. 
Mozaffari et al.~\cite{Moza15} considered the interference
among two UAVs and proposed an optimal placement method 
that maximizes the covered area. Moreover, the authors
extended this method to a multiple-UAV scenario in \cite{Moza16}. 
%
Kalantari et al.~\cite{Kala16} proposed a meta-heuristic-based 
optimal 3D placement algorithm that maximizes the user coverage
using the minimum number of UAVs. 
Furthermore, the work~\cite{Ghan18} proposed a 
Q-learning-based 3D deployment method for maximizing 
the data rate of ground users. 
%
%
However, the above methods are fundamentally centralized and offline approaches. 
Since the 3D placement optimization problem of UAVs has a high degree of freedom and 
the computational cost increases with the number of users 
or deployment area, they are not applicable 
to dynamic scenarios in which the density of users changes temporally. 

Despite the importance of the dynamic scenarios, 
dynamic and distributed UAV-BSs deployment
problems have not been  extensively studied. 
Several recently proposed dynamic deployment methods focused on 
constructing a desired formation for relay networks~\cite{Orfa16}
or target tracking~\cite{Dutt18, Zhao18b}. 
%
Thus, they are not for UAV-BS networks, i.e., 
did not consider ground users or backbone networks. 
Zhao et al.~\cite{Zhao18} proposed 
a heuristic-based distributed motion-control method for UAV-BSs. 
%
However, they focused on the connectivity among UAVs
and only UAVs' horizontal movements. 
They did not consider the probabilistic characteristics of the SINR of users 
including A2G channel and interference. 
In contrast, 
we theoretically formulate the distributed UAV 3D deployment problem 
by considering interference among UAVs and the coverage probability 
of users. 
We also demonstrate that by controlling the UAV altitudes, 
the total coverage  can be efficiently improved. 

%

%

\section{Model Description}\label{sec-model}
\subsection{Network modeling}\label{subsec-network}
In this paper, we consider a downlink network 
in which UAVs act as flying BSs and aim to enhance
the communication quality of users. 
We first provide an overview of the network model. 
As shown in Fig.~\ref{fig:system}, the network consists mainly of 
four types of elements: UAVs, on-ground user equipments
(UEs), on-ground sensors (GSs), and a remote gateway (RG). 
UAVs provide internet connections to the UEs by connecting to the RG, which also 
provides UAVs connections to back-bone 
networks (e.g., satellites and larger UAVs~\cite{Zhao18}). 
%
GSs are assumed to be capable of detecting UEs nearby the GSs
and estimating their (average) number (e.g., Wi-Fi access points with 
RSSI- or CSI-based sensing methods~\cite{Scha14, Xi14}). 
The sensing results observed at GSs 
are periodically reported to a central server 
through a wired or wireless link. 
Then, the central server estimates the entire intensity of the UEs
and distributes it among the UAVs through the RG. 

\begin{figure}[!t]
\centering
\includegraphics[width=3.4in]{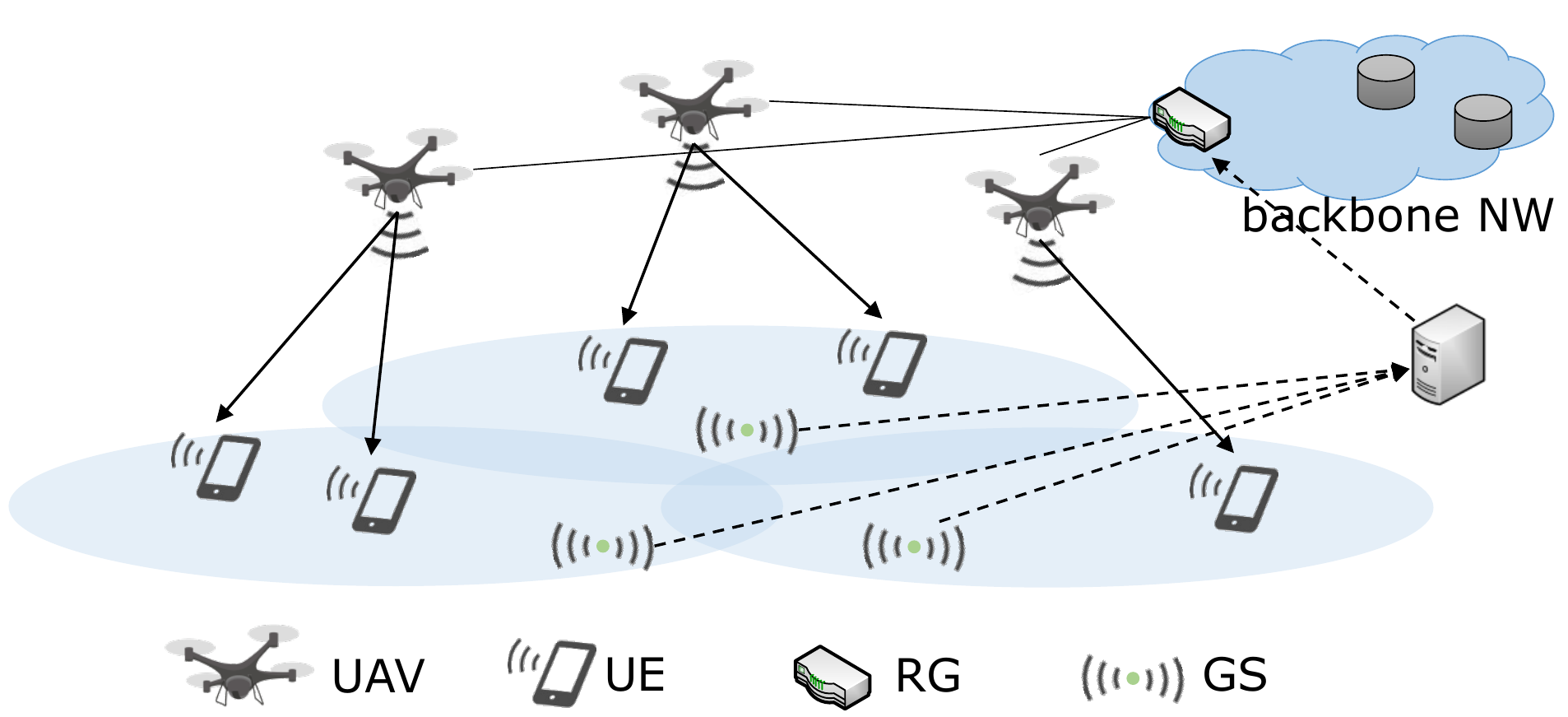}
\caption{Illustration of UAV network model. }
\label{fig:system}
\vspace{-3.0mm}
\end{figure}

We now describe our model in detail. 
Table~\ref{tab:notation} summarizes the notations used in this paper. 
We assume that time is slotted as $t \in \{0,1,\dots\} \triangleq \bbZ_+$. 
Each time slot corresponds to the time at which the UAV positions are updated. 
%
We assume that the UAVs can move anywhere in a certain 
closed convex set $\Omega_3 \subset \bbR^2 \times \bbR_+$ 
where $ \bbR_+ = [0, \infty)$ 
and that the UEs are in a closed convex set $\Omega_2 \subseteq \calP(\Omega_3) $ on the ground. 
Here,  $\calP:\bbR^3 \to \bbR^2$ represents the projection from the 3D space to the plane (ground). 
Let $\vc{u}_i(t) = [u_{i0}(t), u_{i1}(t), u_{i2}(t)]^{\top} \in \Omega_3$ 
($i \in \{1,\dots,U\} \triangleq \calU$) denote the position
of the $i$-th UAV at time $t \in \bbZ_+$. 
Here, $U$ denotes the total number of UAVs. We also write 
$\vc{u}(t) = [\vc{u}^{\top}_1(t), \dots, \vc{u}^{\top}_U(t)]^{\top} \in 
(\Omega_3)^{U}$. 
Furthermore, we assume that the UEs are distributed according to 
a certain (inhomogeneous) spatial point process with  an intensity function
$\lambda(\vc{y})~(\vc{y} \in \Omega_2)$. 
Although $\lambda(\cdot)$ depends on $t$ due to the dynamics of UEs,
we omit the index for simplicity. 

Next, we  describe the activities of GSs. 
Let $\vc{g}_n \in \Omega_2~(n=1,\dots,N)$ denote the positions 
of GSs. Here, $N$ represents the total number of the GSs.
The GSs are assumed to be capable of  detecting and estimating the number of UEs around them by sensing and
report their sensing results periodically  at time $T_k \triangleq k T~(k \in \bbZ_+)$. 
Here, $T$ denotes the sensing period and is assumed to be 
larger than UAVs' update interval (i.e., $T \gg1$)
due to the computational time of sensing or crowd density estimation. 
Furthermore, let $x_n(T_k) \in \bbR_+~(n=1,\dots,N)$ denote the sensing result reported from the $n$-th GS at time $T_k$, 
i.e., the (average) expected number of UEs 
per unit space in the neighborhood of $\vc{g}_n$. 
Thus, it can be regarded as an observation of $\lambda(\vc{g}_n)$ at time $T_k$. 
We also express $\vc{x}(T_k) = [x_1(T_k), \dots, x_N(T_k)]^{\top} \in (\bbR_+)^N$. 
Using $\vc{x}(T_k)$, the crowd density estimation part infers
the entire intensity function 
as $\widetilde{\lambda}_k(\vc{y})$ ($\vc{y} \in \Omega_2$). 
Note that we can easily consider the scenario where UAVs are equipped with sensors
by taking into account sensing results at $\vc{u}_i(t)$'s. 
However, in this study, we assume that only GSs can perform sensing, for simplicity. 

\begin{table}[ht]
\caption{List of Notations}
\label{tab:notation}
\centering
\begin{tabular}{ll} \hline
$\vc{u}_i(t)$, $\vc{u}(t)$ & positions of $i$-th UAV and all UAVs at time $t$ \\
$U$                        & total number of UAVs \\
$\vc{g}_n$                 & position of $n$-th GS \\
$N$                        & total number of GSs \\
$\lambda(\vc{y})$, $\widetilde{\lambda}(\vc{y}) $& 
                            true and estimated intensity of UEs at $\vc{y}$\\
$T$, $T_k$                 & sensing period and $k$-th sensing time\\
$x_n(T_k)$, $\vc{x}(T_k)$  & sensing result  at time $T_k$ from $n$-th or all GSs\\
$\theta_{i,y}(t)$          & elevation angle from $\vc{y}$ to $i$-th UAV at time $t$\\
$d_{i,y}$(t)               & distance between $i$-th UAV and UE at $\vc{y}$ at time $t$\\
$\calV_i(\vc{u})$          & UAV cell of $i$-th UAV\\
$\calN_i(t)$               & neighborhood of $i$-th UAV\\
$h_{i,y}(t)$               & fading gain corresponding to  $i$-th UAV and UE at $\vc{y}$ \\
$I_{i,y}(t)$               & interference at UE at $\vc{y}$ served by $i$-th UAV \\
$\sigma$                   & normalized thermal noise power\\
$\Theta$                   & SINR threshold  \\
 \hline
\end{tabular}
\end{table}

\subsection{Channel modeling and UAV-cell definition}

We consider the following radio propagation model for
UAV--UE channels. Due to the blockage effect in the A2G link, 
the link conditions can be 
divided into LoS and non-LoS (NLoS) conditions. 
We adopt a probabilistic LoS model similar to that in \cite{Hour14, Wu18}, 
wherein a link condition ($\rmL: \LoS$, $\rmN: \NLoS$)
 is determined by the following probability: 
%
%
\begin{align}
\PP(\rmL ; \theta_{i, y}(t))
&= {1 \over 1 + b_1 \exp(- b_0({180 \over \pi}\theta_{i,y}(t) - b_1))},
\label{eq-LoS}
\\
\PP(\rmN ; \theta_{i, y}(t))
&=
1- \PP(\rmL ;\theta_{i, y}(t)),
\label{eq-NLoS}
\end{align}
where  $b_0$ and $b_1$ are constants dependent on 
the environments (e.g., urban and rural) and are provided in \cite{Hour14}, 
and $\theta_{i,y}(t)\in [0, \pi/2]$ is the elevation angle 
from a UE at $\vc{y} \in \Omega_2$ to the $i$-th UAV. Thus, if we denote
$d_{i,y}(t)$ as the distance between the $i$-th UAV and a UE at $\vc{y}$, 
$\theta_{i,y}(t) = \sin^{-1}(u_{i2}(t) /d_{i,y}(t))$. 
For simplicity, we assume that all the  UAVs and UEs are equipped with
omni-directional antennas. 
%
%
Although the channel condition is spatially and temporally 
correlated in general, for simplicity, we assume that they are independently 
determined at each time and position. 
%
We assume a distance-dependent path-loss model~\cite{Wu18, Zhu18} such that,
for a distance $d$, 
\[
\ell_q(d)  = \beta_q (\varepsilon_0 + d)^{-\alpha_q}
\quad
\mbox{for $q \in \{\rmL, \rmN\}$},
\]
where $\alpha_q > 2 $ is a path-loss exponent, 
$\varepsilon_0 > 0$ denotes a constant to avoid
a singularity at $d = 0$,  
and $\beta_q> 0$ is a constant that depends 
on the environment. 
%
%
Furthermore, we model the small-scale fading effect
in UAV-UE channels with the independent Nakagami-$m$
fading similar to that in \cite{Zhu18}. 
Let $h_{i,y}(t)$ denote the small-scale 
fading gain corresponding 
to the channel between the $i$-th UAV and a UE at $\vc{y}$. 
Then, $h_{i,y}(t)$ is distributed according to 
the normalized gamma distribution with a parameter $m_{q}$
$(q \in \{\rmL, \rmN\})$. 
For simplicity, we assume  $m_q$ to be an integer 
throughout this paper. 
In addition, we omit the effects of shadowing
and Doppler shift for mathematical tractability. 
The transmission power of each UAV is identical and 
is assumed to be normalized to one. 

We next define the serving area of each UAV, i.e., a UAV-cell. 
In this paper, we assume that each UAV has a 
capacity sufficient  for serving UEs.
Moreover, each UE is assumed to associate 
with the UAV that provides the strongest average signal power. 
By definition, the average signal power from 
the $i$-th UAV to a UE at $\vc{y}$ is expressed as
\begin{equation}
S_{i,y}(t) = 
\sum_{q \in \{\rmL, \rmN\}} \PP(q; \theta_{i,y}(t))
\ell_q(d_{i,y}(t)).
\label{eq-S_{i,y}}
\end{equation}
The UAV-cell is then defined as
a {\it signal-weighted Voronoi cell}, i.e.,
for the $i$-th UAV, 
\begin{equation}
\calV_i(\vc{u}(t))
=
\left\{
\vc{y} \in \Omega_2 \mid 
S_{i,y}(t)\ge S_{j,y}(t),
\forall j \in \calU\backslash\{i\}
\right\}.
\label{eq-def-calC_i}
\end{equation}
From this expression, we can construct an associated 
{\it signal-weighted Delaunay graph} $\calG(t)$
by choosing the set of vertices as $\calU$
and the set of edges 
as pairs of UAVs whose signal-weighted Voronoi cells are adjacent. 
Let $\calN_i(t)$ denote the set of neighboring 
UAVs of the $i$-th UAV on the graph $\calG(t)$. 
We assume that the $i$-th UAV can communicate with only
$\calN_i(t)$ in our distributed collaborative deployment method.

\subsection{Problem description}

%
In this paper, we consider the {\it total coverage} of UEs,
i.e., the total expected number of covered UEs,
as the performance metric of UAV deployment.
We assume that a UE is {\it covered}
if the SINR at the UE exceeds a threshold $\Theta$. 
%
%
By considering a random channel condition and fading, 
the {\it coverage probability} of a UE at $\vc{y}$ served by the $i$-th UAV is
expressed as
\begin{equation}
C_{i,y}(\vc{u}(t)): = C_{i,y}(\vc{u}(t); \Theta)  = \PP(\SINR_{i,y}(t) > \Theta).
\label{eq-def-A}
\end{equation}
Here, $\SINR_{i,y}(t)$ is expressed as 
\begin{equation}
\SINR_{i,y}(t) = {
h_{i,y}(t) \ell_q(d_{i,y}(t))
\over 
I_{i,y}(t) + \sigma
},
\label{eq-SINR}
\end{equation}
where $I_{i,y}(t)$ is the total interference power
defined as
%
\begin{equation}
I_{i,y}(t) = \sum_{j \in \calU\backslash\{ i\}}
 h_{j,y}(t) \ell_q(d_{i,y}(t)),
\label{eq-I_{ij}}
\end{equation}
and $\sigma$ denotes the thermal noise power and 
is assumed to be constant. 
%
%
By considering the coverage probability, 
we can characterize the communication quality of UEs more flexibly
than the existing performance metrics based on 
mean path-loss \cite{Bor16, Alze18} or average SINR~\cite{Kala16, Zhao18}. 
By taking the expectation of the number
of covered UEs on $\Omega_2$ and applying Campbell's theorem~\cite{Stoy95}, 
the total coverage can be expressed as
%
\begin{align}
&\int_{\Omega_2}
\bbE[\onev(\mbox{UE at $\vc{y}$ is covered})]
\lambda(\vc{y}) \rmd \vc{y}
\nonumber
\\
&~~~~=
\sum_{i \in \calU}
\int_{\calV_i(\vc{u})}
C_{i,y}(\vc{u})
\lambda(\vc{y}) \rmd \vc{y},
\label{eq-true-F}
\end{align}
where $\onev(\cdot)$ denotes the indicator function. 

We are now prepared to state our optimization problem. 
Our objective is to determine an optimal 3D deployment of UAVs 
that maximizes (\ref{eq-true-F}) in a distributed manner. 
Since the true intensity $\lambda(\vc{y})$ cannot be obtained
in general, we use an estimated intensity for the deployment 
optimization instead. 


\begin{prob}\label{prob-opt}
Assume that i) each UAV $i$ $(i \in \calU)$  can 
communicate only with its neighborhood $\calN_i(t)$ on $\calG(t)$
and ii) the estimated intensity at $T_k$ ($k \in \bbZ_+$) 
based on sensing result $\vc{x}(T_k)$
is given by $\widetilde{\lambda}_k(\vc{y})$. 
For each sensing period $k$, determine the optimal UAV deployment $\vc{u}\in (\Omega_3)^U$ 
that maximizes the estimated total coverage $F(\vc{u} \mid \vc{x}(T_k))$ defined as
%
%
\begin{align}
F(\vc{u} \mid \vc{x}(T_k))
&=
\sum_{i\in \calU}
\int_{\calV_i(\vc{u})}
C_{i,y}(\vc{u})
\widetilde{\lambda}_k(\vc{y})
\rmd \vc{y}.
\label{eq-J(u)-true}
\end{align}
\end{prob}

%
%

%
%
%

\section{Distributed UAV 3D Deployment Method}\label{sec-main}

In this section, we present 
the distributed collaborative UAV 3D deployment method. 
Our method mainly consists
of two parts: i) sensing-aided crowd density estimation; and
ii) distributed push-sum algorithm. 
At each sensing time $T_k$, GSs send $\vc{x}(T_k)$ to a central server. 
The crowd estimation part then estimates 
$\lambda(\vc{y})$. 
During the $k$-th sensing period, UAVs update their 3D positions
by using $\widetilde{\lambda}_k(\vc{y})$ 
via the distributed push-sum algorithm, which is aimed 
at solving
Problem~\ref{prob-opt}. 

This section is organized as follows. 
We first analyze our performance metrics, namely
the coverage probability and total coverage 
in Section~\ref{subsec-avai}. 
We then describe the crowd density estimation 
part in Section~\ref{subsec-crowd} and
present the distributed push-sum algorithm 
part in Section~\ref{subsec-pushsum}. 
Finally, we prove the convergence of the distributed algorithm
in Section~\ref{subsec-conv}.

\subsection{Coverage analysis}\label{subsec-avai}
%
Since our distributed optimization method is gradient-based, 
we analyze the total coverage and derive its gradient in this section. 
We omit the index $t$ in this section because  a fixed time is being considered. 
%
We start by deriving 
the theoretical expression of the coverage probability of  UEs. 
By using the expressions of (\ref{eq-def-A}) and (\ref{eq-SINR}), 
we can obtain the following result.  A sketch of the proof
is provided in  Appendix~A. 
%
\begin{lem}\label{lem-A}
If the deployment of UAVs is $\vc{u} \in (\Omega_3)^U$, 
the coverage probability $C_{i,y}(\vc{u})$ 
can be approximated by
%
\begin{align}
C_{i,y}(\vc{u}) 
\approx
\widetilde{C}_{i,y}(\vc{u}) 
&\triangleq
\sum_{q_0 \in \{\rmL, \rmN\}}
\PP(q_0; \theta_{i,y})
\sum_{k=1}^{m_{q_0}}(-1)^{k+1}
\nonumber
\\
&\times
{m_q \choose k}
e^{- k \sigma \gamma_{q_0, i,y}}
\calL_{I_{i,y}}(k\gamma_{q_0, i,y}), 
\label{eq-A_{i,n}}
\end{align}
where 
%
$\gamma_{q, i,y}
= 
{\eta_q \Theta \over \ell_q(d_{i,y})}
$
with $\eta_{q} = m_q(m_q!)^{-{1\over m_q}}$
for $q \in \{\rmL, \rmN\}$,
and $\calL_{I_{i,y}}(s)$ is given by 
\begin{align}
\calL_{I_{i,y}}(s)=\!\!
\prod_{j \in \calU\backslash\{i\}}\!
\sum_{q_j \in \{\rmN,\rmL \}}\!\!\!
\PP(q_j; \theta_{j,y})
\left(
1 + {s  \ell_{q_j}(d_{j,y}) \over m_{q_j}}
\right)^{-m_{q_j}}\!\!\!\!\!\!\!\!.
\nonumber
\end{align}
\end{lem}
%
%
%
On the basis of Lemma~\ref{lem-A}, 
we aim to solve Problem~\ref{prob-opt}
through $\widetilde{C}_{i,y}(\vc{u})$
in what follows. 
In other words, we replace $C_{i,y}(\vc{u})$ in the 
objective function in (\ref{eq-J(u)-true})
by $\widetilde{C}_{i,y}(\vc{u})$ and consider 
the optimization problem of the following function:
%
\begin{align}
\widetilde{F}(\vc{u} \mid \vc{x}(T_k))
&=
\sum_{i\in \calU}
\int_{\calV_i(\vc{u})}
\widetilde{C}_{i,y}(\vc{u})
\widetilde{\lambda}_k(\vc{y})
\rmd \vc{y}.
\label{eq-J(u)}
\end{align}
%

Furthermore, from (\ref{eq-A_{i,n}}), we can calculate
the derivatives of $\widetilde{C}_{i,y}(\vc{u})$
with respect to $\vc{u}_j$, i.e., 
\[
\widetilde{\vc{c}}^j_{i,y}(\vc{u}) \triangleq 
{\rmd \over \rmd \vc{u}_j}\widetilde{C}_{i,y}(\vc{u}),
\qquad
i,j \in \calU.
\vspace{-1mm}
\]
Since we can conveniently obtain $\widetilde{\vc{c}}^j_{i,y}(\vc{u})$
from Lemma~\ref{lem-A}, we omit their explicit 
expressions due to space limitations.

Next, we demonstrate that the derivative of $\widetilde{F}(\vc{u}\mid\vc{x})$ with respect
to $\vc{u}_j$ can be approximated through $\widetilde{\vc{c}}_{i,y}^j(\vc{u})$. 
A sketch of the proof  is provided in Appendix~B.
In what follows, we simply write $\widetilde{F}(\vc{u}):= \widetilde{F}(\vc{u}\mid \vc{x})$ for readability.  

\begin{lem}\label{lem-div-F}
If the estimated intensity function  
is $\widetilde{\lambda}(\vc{y})$ ($\vc{y} \in \Omega_2$), 
%
\begin{equation}
{\rmd  \widetilde{F}(\vc{u} )\over \rmd \vc{u}_j}
\approx
\sum_{i \in \calU}
\int_{\calV_i(\vc{u})} 
\widetilde{\vc{c}}_{i,y}^j(\vc{u})
\widetilde{\lambda}(\vc{y})
\rmd \vc{y},
\quad
\mbox{$j \in \calU$}.
\label{eq-div-F-u_j}
\end{equation}
%
\end{lem}
%

According to Lemma~\ref{lem-div-F}, 
the gradient of $\widetilde{F}(\vc{u})$ 
exhibits the following useful (approximate) separation property, 
which enables us to optimize $\widetilde{F}(\vc{u})$
in a distributed manner: 
\begin{align}
\nabla \widetilde{F}(\vc{u} ) &\approx \sum_{i\in\calU} \widehat{\vc{f}}_i(\vc{u}),
\label{eq-nabla-F=sum-tilde-f_i}
\end{align}
where, for $i \in \calU$, 
\begin{align}
\widehat{\vc{f}}_i(\vc{u} ) 
&\triangleq
\left[
\int_{\calV_i(\vc{u})} 
\widetilde{\vc{c}}_{i,y}^j(\vc{u})
\widetilde{\lambda}(\vc{y})
\rmd \vc{y}
\right]^\top_{j \in \calU}
\nonumber
\\
&=
\int_{\calV_i(\vc{u})} 
\nabla \widetilde{C}_{i,y}(\vc{u})
\widetilde{\lambda}(\vc{y})
\rmd \vc{y}. 
\label{eq-def-tilde-f_i}
\end{align}
The above relationship 
indicates that there exists $\widehat{F}(\vc{u})$
such that $\nabla \widehat{F}(\vc{u}) = 
\sum_{i\in\calU} \widehat{\vc{f}}_i(\vc{u}) $ and
$\widetilde{F}(\vc{u}) \approx \widehat{F}(\vc{u})$. 
%

%

%

\subsection{Crowd density estimation}\label{subsec-crowd}
We next describe the crowd density estimation. 
The aim of this part is to estimate the 
intensity function 
from the sensing results at sensing time $T_k$ ($k \in \bbZ_+$),
i.e., to obtain $\widetilde{\lambda}_k(\vc{y})$ from $\vc{x}(T_k)$. 
%
Typically, the  density of UEs is spatially correlated
owing to geographical factors, such as roads and attractions. 
Thus, by modeling this spatial correlation,  
we infer the entire intensity function  with a limited number of GSs.

For a fixed sensing time, 
we consider a spatial Gaussian process (GP) prior
$\GP(\mu, k_0(\vc{y}, \vc{y}'))$ 
$(\vc{y}, \vc{y}' \in \bbR^2)$ 
for  $\lambda(\vc{y})$. 
Here, the constant $\mu$ expresses the mean 
of $\GP$ and corresponds to a prior expectation of
$\lambda(\vc{y})$. 
Meanwhile, $k_0(\vc{y}, \vc{y}')$ 
is the kernel function of $\GP$ and
represents the spatial correlation of $\lambda(\vc{y})$. 
In this paper, we adopt a well-accepted Gaussian kernel: 
%
\begin{equation}
k_0(\vc{y}, \vc{y}') = A_0
 \exp\left(- {\|\vc{y} - \vc{y}'\|^2 \over A_1}\right),
\qquad \vc{y}, \vc{y}'\in \bbR^2, 
\label{eq-G-kernel}
\end{equation}
where $\|\cdot\|$ denotes the Euclidean norm
and $A_0$ and $A_1$ are constants. 
We assume that the parameters of GP (i.e., $\mu$, 
$A_0$, and $A_1$) are obtained from statistical data. 
GPs are widely used in modeling various spatial
data~\cite{Rasm06}
ranging from geology to environmental sciences. 
According to an established  property of a GP, 
if we regard observations $\vc{x}(T_k)$ 
as realizations of $\lambda(\vc{g}_n)$ ($n=1,\dots,N$) at $T_k$, 
they are 
distributed with a multivariate Gaussian
distribution. 
%
Furthermore, the posterior distribution of $\lambda(\vc{y})$
given $\vc{x}(T_k)$ also becomes a Gaussian (e.g., \cite{Rasm06}): 
as follows: 
\begin{align}
&\lambda(\vc{y}) \mid \vc{x}(T_k), \{\vc{g}_n; n=1,\dots,N\}
\nonumber
\\
&\quad\sim
\calN(\mu_y(\vc{x}(T_k)),
k_0(\vc{y}, \vc{y}) - \vc{k}_y^\top\vc{K}^{-1}_{\!g}\vc{k}_y),
\end{align}
where
\begin{align}
\mu_{y}(\vc{x}(T_k)) &= \mu + \vc{k}_y^\top \vc{K}_{\! g}^{-1} (\vc{x}(T_k) - \mu\vc{e}),
\label{eq-mu_y}
\\
\vc{k}_y &= [k_0(\vc{g}_1, \vc{y}), \dots, k_0(\vc{g}_N, \vc{y})]^{\top},
\nonumber
\\
[K_{\!g}]_{n,n'} &= k_0(\vc{g}_n, \vc{g}_n'), \qquad  n,n' =1,\dots,N,
\nonumber
\end{align}
and $\vc{e}$ denotes a vector of ones with 
an appropriate dimension. 
As a result, UAVs can estimate the entire intensity from 
the observations $\vc{x}(T_k)$ by setting
\begin{equation}
\widetilde{\lambda}_k(\vc{y}) = \mu_y(\vc{x}(T_k)), \qquad \vc{y} \in \Omega_2. 
\label{eq-lambda-mu}
\end{equation}
By substituting (\ref{eq-lambda-mu}) into (\ref{eq-J(u)}), 
$\widetilde{F}(\vc{u}\mid \vc{x}(T_k))$ in (\ref{eq-J(u)}) is now rewritten as
\begin{align}
\widetilde{F}(\vc{u} \mid \vc{x}(T_k))
&=
\sum_{i\in \calU}
\int_{\calV_i(\vc{u})} 
\widetilde{C}_{i,y}(\vc{u}) 
\mu_y(\vc{x}(T_k))
\rmd \vc{y}. 
\label{eq-J(u)-2}
\end{align}

\subsection{Distributed push-sum algorithm}\label{subsec-pushsum}
In this section, we describe the distributed push-sum algorithm. 
In this part, UAVs update their positions and
maximize  (\ref{eq-J(u)-2})
in a distributed manner. 
Our objective function (\ref{eq-J(u)-2}) depends on the positions of all the UAVs. 
Therefore, it is not trivial to solve the problem in a distributed
manner. However, the useful separation property of 
$\nabla \widetilde{F}(\vc{u})$ in 
(\ref{eq-nabla-F=sum-tilde-f_i})
enables us to apply a distributed push-sum
framework. This framework was first proposed by \cite{Kemp03}
for average computation. It was recently extended to 
a time varying network by \cite{Nedi15}, 
and to a non-convex optimization problem by Tatarenko and Touri~\cite{Tata17}. 

%

To state the  algorithm, we first introduce
an extension $\widehat{\vc{f}}^\ast_i$ 
of $\widehat{\vc{f}}_i$ that 
is defined on $(\bbR^3)^U$ and satisfies 
\begin{equation}
\widehat{\vc{f}}^\ast_i(\vc{u}) = 
\left\{
\begin{array}{ll}
\widehat{\vc{f}}_i(\vc{u}), & \vc{u} \in (\Omega_3)^U,\\
J \breve{\vc{u}}, & \vc{u} \in \overline{D},\\
\end{array}
\right.
\label{eq-f_i-extend}
\end{equation}
where $\overline{D}$ denotes the complement of 
an open set $D \subset (\bbR^3)^U$ disjoint with $(\Omega_3)^U$. 
In addition, $J> 0$ is a constant 
and $\breve{\vc{u}}$ denotes a unit vector directed from 
$\vc{u}$ toward its closest point on the boundary of $(\Omega_3)^U$. 
In the next section, we prove the existence of $\widehat{\vc{f}}^\ast_i$
that satisfies several conditions required for the convergence
of the distributed push-sum algorithm. 


\begin{algorithm}[t]                      
\caption{Distributed push-sum algorithm}         
\label{alg1}  
\begin{algorithmic}[1]
\STATE{Initialize: $\vc{u}^{[i]}(0)$, $\vc{w}^{[i]}(0)$, $\vc{\xi}^{[i]}(0)$ and $\phi^{[i]}(0)$; }
\FOR{$t = 0,1,\dots$}
\FOR{each UAV $i$} 
\STATE{Update $\calN_i$ and $\calV_i(\vc{u})$ of UAV $i$;}
\FOR{each neighboring UAV $j$} 
\STATE{
Send  $\vc{\xi}^{[i]}(t)$ and $\phi^{[i]}(t)$ to UAV $j$; and\\
Receive  $\vc{\xi}^{[j]}(t)$ and $\phi^{[j]}(t)$ from UAV $j$;
}
\ENDFOR
\STATE{
Update $\vc{u}^{[i]}(t)$, $\vc{w}^{[i]}(t)$, $\vc{\xi}^{[i]}(t)$ and $\phi^{[i]}(t)$ by (\ref{eq-update-u_i}); 
}
\STATE{
Update position of UAV $i$: $\vc{u}_i(t + 1) := \vc{u}_i^{[i]}(t+1)$;
}
\ENDFOR
\ENDFOR
\end{algorithmic}
%
\end{algorithm}

We now describe the algorithm in detail. 
The pseudo-code is given in Algorithm~\ref{alg1}. 
Each UAV $i$ (i.e., the $i$-th UAV) maintains a $3U$-dimensional vector
$\vc{u}^{[i]}(t) \in (\bbR^3)^U$, which represents the {\it pseudo} positions
of the UAVs. Here, {\it pseudo} indicates that 
this location information is {\it local} for the UAV $i$
and does not always represent the exact positions of the other UAVs
(only $\vc{u}^{[i]}_i(t)$ expresses the exact position of the UAV $i$).  
This is because each UAV can only communicate with its neighboring UAVs.
They also maintain auxiliary
$3U$-dimensional vectors $\vc{w}^{[i]}(t)$, 
and $\vc{\xi}^{[i]}(t)$, and a scholar variable $\phi^{[i]}(t)$. 
In each time step,  each UAV $i$  updates its local variables
according to the following rules:
\begin{align}
\vc{w}^{[i]}(t+1) &:= \!\!\! \sum_{j \in \calN_i(t)} {\vc{\xi}^{[j]}(t) \over d_j(t)},~
\phi^{[i]}(t+1) := \!\!\! \sum_{j \in \calN_i(t)} {\phi^{[j]}(t) \over d_j(t)},
\nonumber
\\
\vc{u}^{[i]}(t+1) &:= {\vc{w}^{[i]}(t+1) \over \phi^{[i]}(t+1)}, 
\label{eq-update-u_i}
\\
\vc{\xi}^{[i]}(t+1) &:= \vc{w}^{[i]}(t+1) 
\nonumber
\\
&~~~+ a(t+1)(\widehat{\vc{f}}^\ast_i(\vc{u}^{[i]}(t+1)) + \vc{\kappa}^{[i]}(t+1)),
\nonumber
\end{align}
where  $\vc{\kappa}^{[i]}(t) \in (\bbR^3)^U$ 
are i.i.d. random vectors whose entries are
independent random variables with zero mean and unit variance
for all $t \in \bbZ_+$. 
%
Moreover, $\{a(t)\}$ is a positive non-increasing step-size
sequence such that $a(t) = O( {1 \over t^\nu})$ $(\nu \in ({1\over2}, 1))$
and $d_i(t)$ is the node degree of the UAV $i$ on $\calG(t)$.
%

The update rules in (\ref{eq-update-u_i}) can be 
interpreted simply as follows. Each UAV aims to reach a consensus among the UAVs
by exchanging its local information with its neighbors 
and calculating their weighted sum. 
Furthermore, by combining with a gradient descent based on 
its local information (i.e., $\widehat{\vc{f}}^\ast_i(\vc{u})$),
the consensus point is 
steered towards a critical point of the objective function. 
In addition, the random perturbation $\vc{\kappa}^{[i]}(t)$ ensures 
the convergence to a local optimum~\cite{Tata17}.

\subsection{Proof of convergence}\label{subsec-conv}
In this section, we prove that the distributed push-sum
algorithm given in (\ref{eq-update-u_i}) converges to 
a local optimum of the objective function among all the UAVs
by employing the framework in \cite{Tata17}. 
To do this, we first prove 
the existence of an extension 
$\widehat{\vc{f}}^\ast_i(\vc{u})$ $(i \in \calU)$
in (\ref{eq-f_i-extend}) and the antiderivative
$\widehat{F}^\ast(\vc{u})$ of 
$\sum_{i \in \calU} \widehat{\vc{f}}^\ast_i(\vc{u})$
satisfying several regularity properties. 

\begin{lem}\label{lem-F^ast}
There exist functions  
$\widehat{\vc{f}}^\ast_i(\vc{u})$ $(i \in \calU)$
and $\widehat{F}^\ast(\vc{u})$
on $(\bbR^3)^U$ such that 
$\widehat{\vc{f}}^\ast_i(\vc{u}) = \widehat{\vc{f}}_i(\vc{u})$ and 
$\widehat{F}^\ast(\vc{u}) \approx \widetilde{F}(\vc{u})$ 
in $(\Omega^3)^U$ and 
\begin{enumerate}
\item $\widehat{\vc{f}}^\ast_i(\vc{u})$ is bounded on $(\bbR^3)^U$
for all $i \in \cal U$; 
\item $\widehat{\vc{f}}^\ast_i(\vc{u})$ is Lipchitz continuous on $(\bbR^3)^U$
for all $i \in \cal U$; 
\item $\nabla \widehat{F}^\ast(\vc{u}) = \sum_{i \in \calU} \widehat{\vc{f}}^\ast_i(\vc{u})$; and
\item $\lim_{\|\vc{u}\|\to\infty}\widehat{F}^\ast(\vc{u}) = - \infty$. 
\end{enumerate}
\end{lem}
\begin{proof}
Due to space limitations, we provide only  a summary of the proof. 
In what follows, we fix $i \in \calU$ arbitrary. 
We can confirm that 
$\widetilde{\vc{c}}^j_{i,y}(\vc{u})$ $(j\in \calU)$ are all 
bounded on $(\Omega_3)^U$. 
We can also confirm that $\vc{\mu}_y(\vc{x})$ is bounded on $\vc{y} \in \Omega_2$
(see (\ref{eq-mu_y})) and $\calV_i(\vc{u})$ is a closed set. 
It then follows from (\ref{eq-def-tilde-f_i}) and (\ref{eq-lambda-mu}) that 
$\widehat{\vc{f}}_i(\vc{u})$ 
is bounded on $(\Omega_3)^U$. Furthermore, 
we can prove that 
$\widehat{\vc{f}}_i(\vc{u})$ is continuously differentiable 
in $(\Omega_3)^U$ and 
the derivatives are bounded on $(\Omega_3)^U$,
which indicates that $\widehat{\vc{f}}_i(\vc{u})$ is Lipschitz continuous 
with a certain Lipschitz constant $L_0$. 
Thus, if we extend the domain of $\widehat{\vc{f}}_i(\vc{u})$ 
to $(\Omega_3)^U \cup \overline{D}$ similarly as in  (\ref{eq-f_i-extend}), 
$\widehat{\vc{f}}_i(\vc{u})$ is bounded and Lipschitz continuous with $L_0$
on  $(\Omega_3)^U \cup \overline{D}$. 
It thus follows from Tietze's extension theorem that 
there exists a bounded and Lipschitz continuous function 
 $\widehat{\vc{f}}^\ast_i(\vc{u})$ on $(\bbR^3)^U$ with $L_0$ 
such that $\widehat{\vc{f}}^\ast_i(\vc{u}) = \widehat{\vc{f}}_i(\vc{u})$
in $(\Omega_3)^U \cup \overline{D}$. 
These results indicate that the statements (i) and (ii) hold.  
Furthermore, by using $\sum_{i\in\calU}\widehat{\vc{f}}^\ast_i(\vc{u})$
and (\ref{eq-nabla-F=sum-tilde-f_i}), we can construct
a continuously differentiable function 
$\widehat{F}^\ast(\vc{u})$ on $(\bbR^3)^U$ such that 
$\widehat{F}^\ast(\vc{u}) \approx \widetilde{F}(\vc{u})$ in $(\Omega^3)^U$ and 
the statements (iii) and (iv) hold,
which completes the proof. 
%
\end{proof}

%
%
%

%

Combining Lemma~\ref{lem-F^ast} with the facts that 
$\widehat{\vc{f}}_i(\vc{u})$ is differentiable in $(\Omega_3)^U$
and the sequence $\{\calG(t)\}$ is strongly connected, i.e., 
the union of their edge sets  is strongly connected, 
we can apply Theorem~5 in \cite{Tata17} 
to $\widehat{F}^\ast(\vc{u})\ (\approx 
\widetilde{F}(\vc{u}))$ and $\widehat{\vc{f}}^\ast_i(\vc{u})$
 in Lemma~\ref{lem-F^ast}
and prove the convergence of the algorithm. 
\begin{thm}
Each $\vc{u}^{[i]}(t)$ $(i \in \calU)$ and 
the average state variable ${1 \over U} \sum_{i=1}^U \vc{\xi}^{[i]}(t)$
in the distributed push-sum algorithm in (\ref{eq-update-u_i}) 
converge to a point in the set of local maxima of 
$\widehat{F}^\ast(\vc{u})$
from any initial state. 
\end{thm}

\section{Simulation Results}\label{sec-simu}
We evaluate our distributed 3D UAV deployment method 
by conducting numerical simulations. 
In each simulation, 
we assume that $U=9$ UAVs are deployed over a ($5000$ m)${}^2$ area ($\Omega_2$)
with a maximum altitude 1500~m. 
We adopted values similar to those in
\cite{Bai15, Zhu18} for the parameters of the channel model
associated with the $\LoS$ and $\NLoS$ conditions. 
We also considered sub-urban, urban, and dense-urban scenarios by 
selecting the values of $b_0$ and $b_1$ presented in \cite{Bor16}. 
The other parameters are listed in Table~\ref{tab:param} unless otherwise stated. 
In addition, we randomly generated spatially correlated
intensity by the following procedure. 
We first divided the whole area
into $50 \times 50 $ small grids, and then sampled the values of $\lambda(\vc{y})$
in each grid from a multivariate
Gaussian distribution whose covariance matrix is determined
by the kernel function in (\ref{eq-G-kernel}).


\begin{table}[ht]
\caption{Simulation Parameters}
\centering
\begin{tabular}{p{4em}p{10.580em} |p{6.55em}p{5.2em} } \hline
\textbf{Parameter}                & \textbf{Value}   & \textbf{Parameter}              &\textbf{Value}\\\hline\hline
$b_0$, $b_1$             &  0.16, 9.61 (urban)    & $\beta_{\rmL}$, $\beta_{\rmN}$   & 0.092, 0.035 \\\cline{3-4}
                         &  0.43, 4.88 (sub urban)  & $m_{\rmL}$, $m_{\rmN}$           & 3, 2    \\\cline{3-4}
           &  0.11, 12.08 (dense urban)                      &    $\alpha_{\rmL}$, $\alpha_{\rmN}$ & 2, 3    \\\hline
$\mu$                           &  100~m${}^{-2}$             & $\sigma$     & $-70$~dBm\\\hline
$A_0$, $A_1$ & 10,  (500)$^2$                &  length of area  & 5000~m\\\hline  
$\Theta$       &  0  dB           &  default altitude  & 200~m\\\hline  
\end{tabular}
\label{tab:param}
\end{table}


%
\subsection{Effectiveness of distributed push sum algorithm}\label{subsec-obj}
%


\begin{figure*}[t]
\centering
\begin{minipage}[t]{0.675\hsize}
\centering
\includegraphics[width=5in]{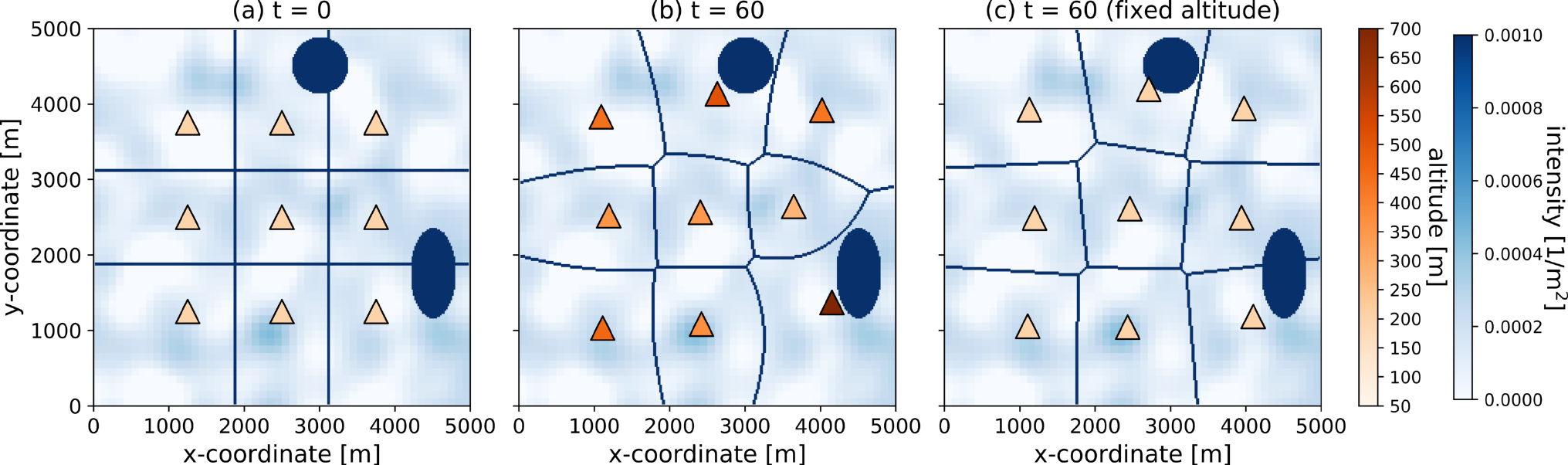}
\caption{3D UAV deployment in urban scenario
(a) $t = 0$ and (b) $t=60$, and (c) fixed altitude scenario $t= 60$.  
Triangles represent horizontal positions of UAVs and their
color depth represent altitudes. Depth of background color expresses 
intensity of UEs and dark ellipses represent artificially added hotspots. 
}
\label{fig:exp-1-U-Theta0_step}
\end{minipage}
~
\begin{minipage}[t]{0.305\hsize}
\centering
\includegraphics[width=2.0in]{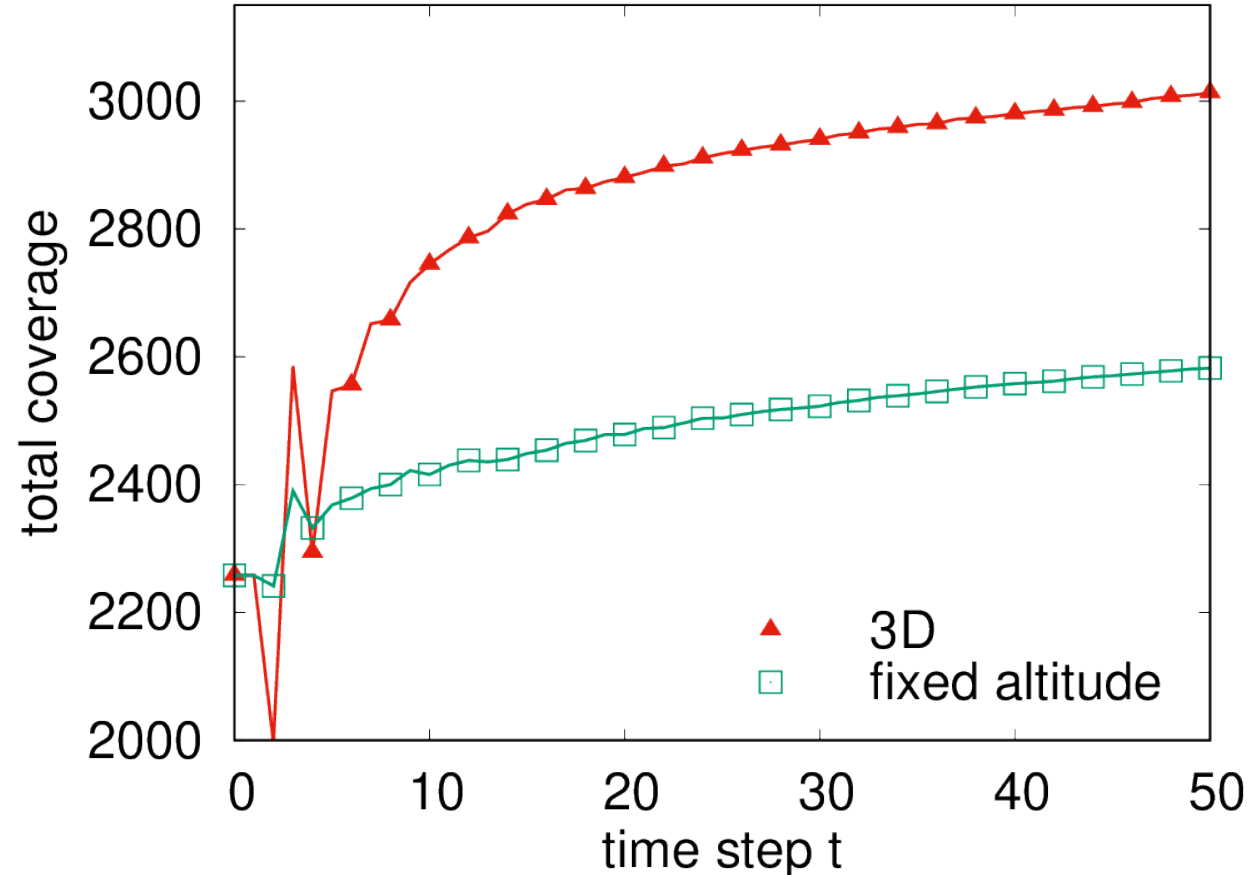}
\caption{Comparison of process of total coverage
with 3D deployment and fixed altitude scenarios. }
\label{fig:exp1-Theta0_alt}
\end{minipage}
\vskip -12pt
\end{figure*}

\begin{figure}[!t]
\centering
\begin{minipage}[t]{0.492\hsize}
\centering
\includegraphics[width=1.75in]{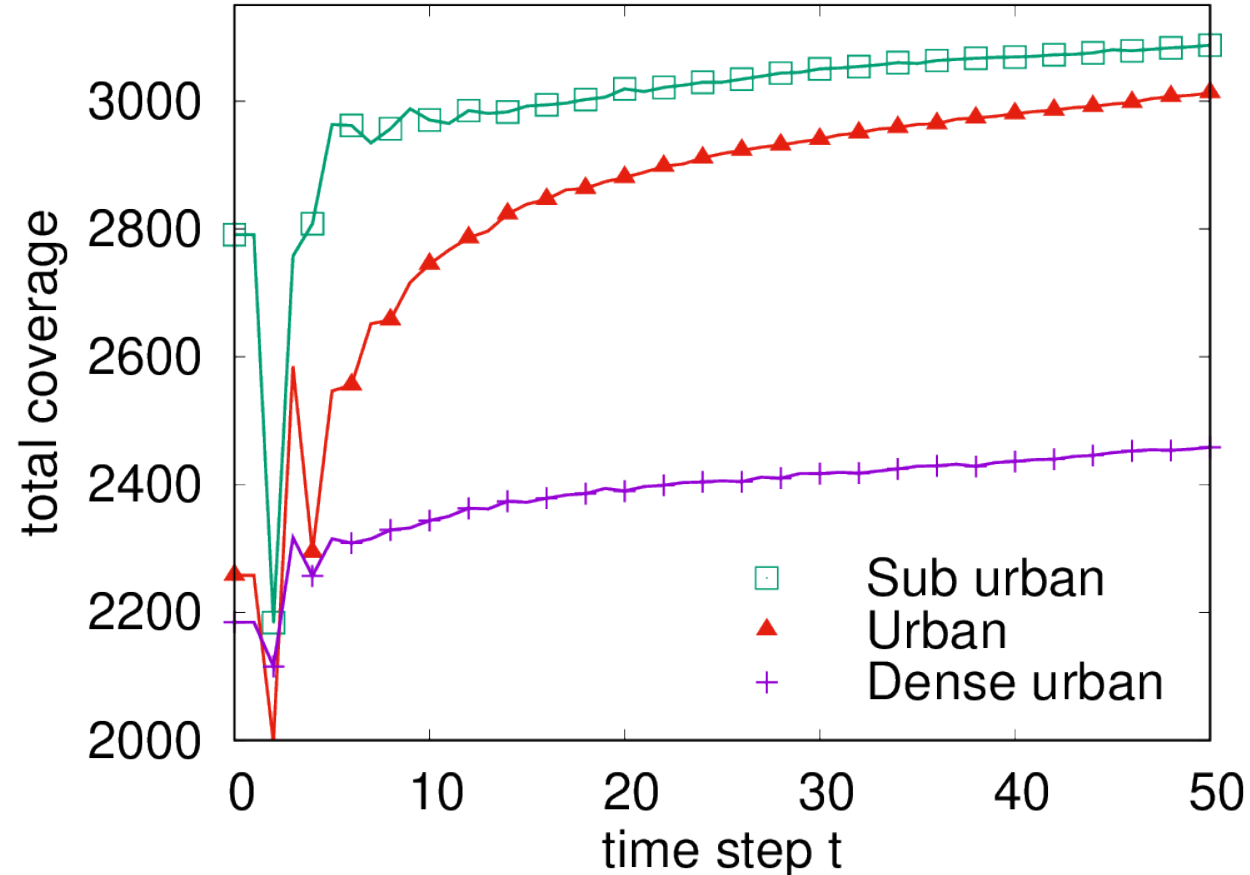}
\caption{Comparison of process of total coverage with different scenarios. }
\label{fig:exp1-Theta0}
\end{minipage}
\begin{minipage}[t]{0.492\hsize}
\centering
\includegraphics[width=1.75in]{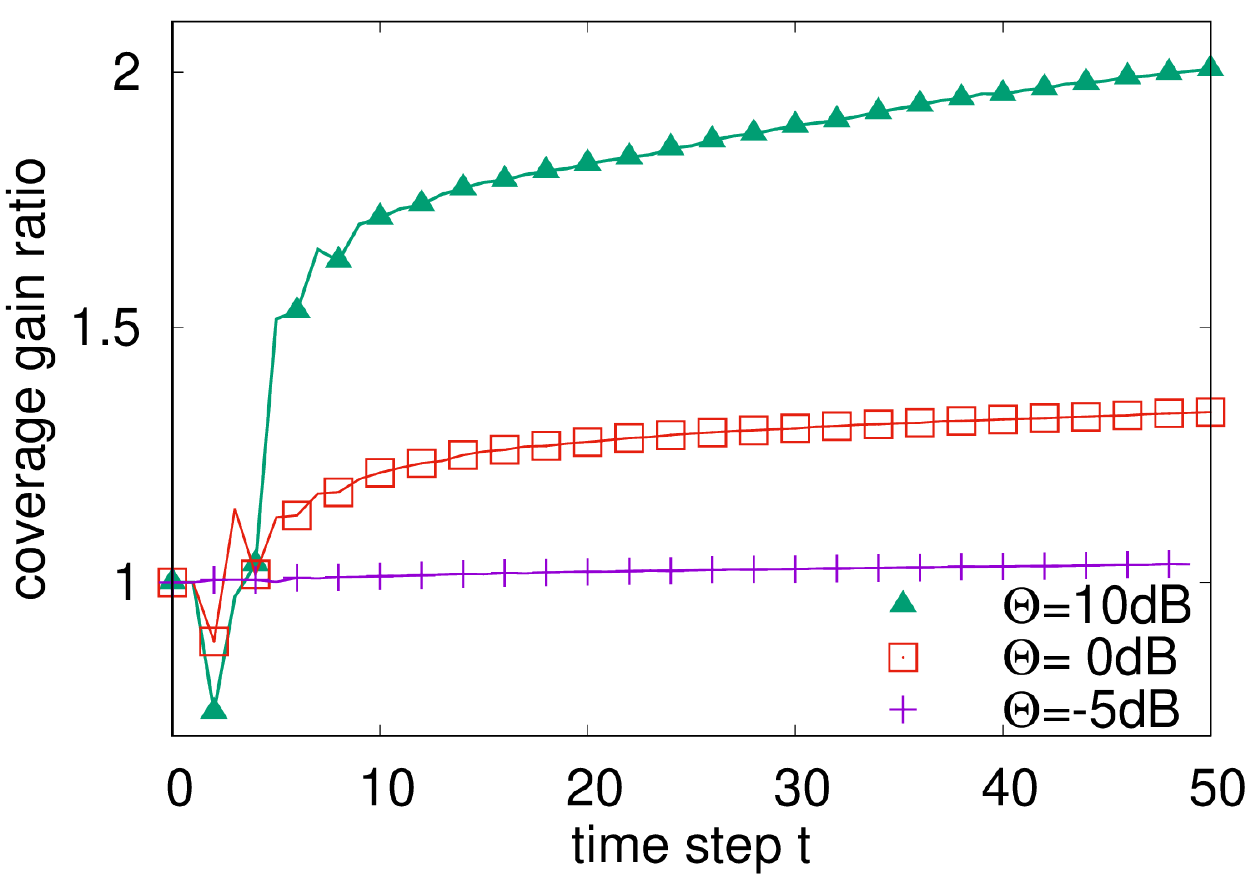}
\caption{Comparison of process of ratio of coverages
with different $\Theta$.}
\label{fig:exp1-U}
\end{minipage}
\vskip -7pt
\end{figure}

%

\begin{figure}[!t]
\centering
\includegraphics[width=3.5in]{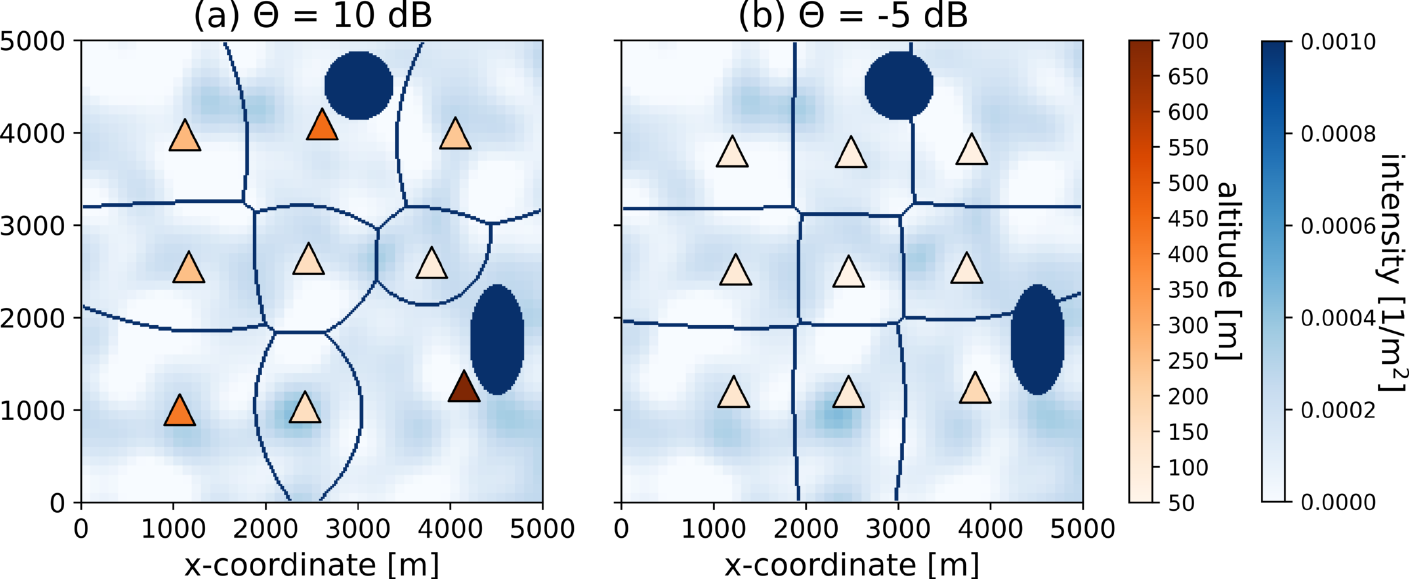}
\vskip -7pt
\caption{
3D UAV deployment in urban scenario at $t = 60$.
}
\label{fig:exp-Theta}
\vskip -7pt
\end{figure}

\begin{figure}[!t]
\centering
\includegraphics[width=3.5in]{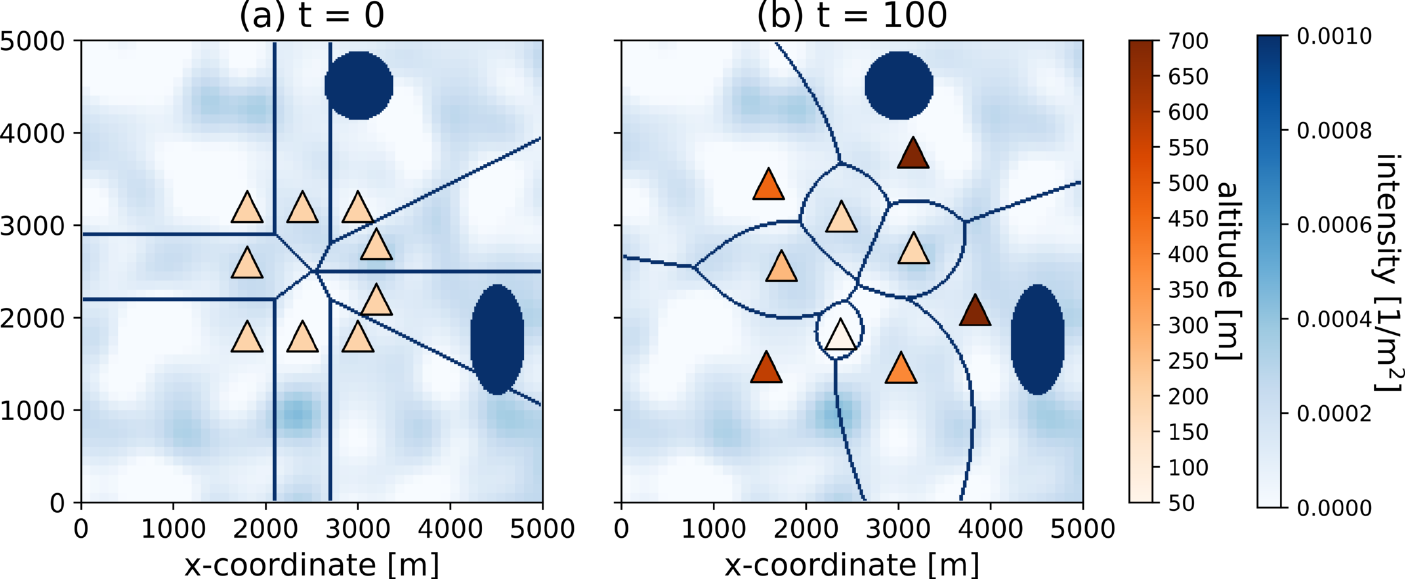}
\vskip -7pt
\caption{
3D UAV deployment starting from centered positions with $\Theta = 0$ dB 
at (a) $t = 0$ and (b) $t = 100$.}
\label{fig:exp-1-p}
\vskip -7pt
\end{figure}

%
We first focus on 
the performance of the distributed push-sum algorithm part
and demonstrate its effectiveness. We also investigate the impacts of
several parameters on the performance and the resulting UAV deployment. 
To achieve this, we focused on a sensing period, i.e., 
a fixed $T_k$, and assumed that all the UAVs knew the true $\lambda(\vc{y})$, 
i.e., $\lambda(\vc{y}) \equiv \widetilde{\lambda}_k(\vc{y})$
($\vc{y} \in \Omega_2$). 
We then set artificial hotspots, in which the intensity 
of UEs is $10\cdot\mu$ 
(see Fig.~\ref{fig:exp-1-U-Theta0_step} (a)). 
In addition, the initial positions of the UAVs were uniformly arranged with 
the equal intervals. 

Fig.~\ref{fig:exp-1-U-Theta0_step} shows the 3D placement of the UAVs
at $t = 0, 60$ in the urban scenario.
We also plotted in Fig.~\ref{fig:exp-1-U-Theta0_step} (c) 
the case where the altitudes of the UAVs
were fixed at the default altitude
by controlling only their horizontal movements. %
The triangles in the graphs express 
the positions of the UAVs, and their altitudes are represented 
by their color depth. Each cell corresponds to each UAV-cell. 
The background color depth expresses 
the intensity of UEs, and the dark ellipses correspond to 
the artificially added hotspots. 
We also illustrate the process of the total coverage 
(i.e., the objective function) 
in Fig.~\ref{fig:exp1-Theta0_alt}. 
%
Fig.~\ref{fig:exp-1-U-Theta0_step} (a) and (b) reveal that 
the 3D positions of the UAVs and associated shapes of the UAV-cells
dynamically adjusted to the hotspots 
so that each hotspot was fitted in a UAV-cell to avoid interference. 
As a result, the total coverage was significantly improved from the initial
deployment (approximately 35 \% higher). 
On the other hand, the total coverage was less improved 
in the fixed altitude scenario. 
Fig.~\ref{fig:exp-1-U-Theta0_step} (b) also shows that 
the altitudes of the UAVs serving hotspots were higher than
those of the other UAVs. Moreover, the altitudes of the UAVs changed
more aggressively than their horizontal positions. 
This is because as the altitude of a UAV
increases, the LoS probability increases,
and thus the UAV-cell and the coverage probability
also increase. 
These results indicate that the total coverage
is more sensitive to the vertical movements of the UAVs
than to their  horizontal ones. 
Thus, their 3D position must be considered for
an efficient UAV deployment optimization. 

Fig.~\ref{fig:exp1-Theta0} shows the 
process of the total coverage in 
sub-urban, urban, and dense urban scenarios. 
Although the total coverage increased in all cases, 
their improvements varied depending on the scenarios.
The reason for this can be considered as follows. 
Since the coverage area of a UAV is 
larger in the sub-urban scenario than the other scenarios~\cite{Hour14, Bor16}, 
the whole area was sufficiently covered even with
an initial placement. Thus, the benefit of our method
was smaller in the sub-urban scenario than the urban one. On the other hand, 
in the dense-urban scenario, ground UEs could not be covered by $U = 9$ UAV-BSs
due to the small coverage area. 

%


We next investigate the impacts of the SINR threshold $\Theta$. 
To see the difference more clearly, 
we calculated the ratio of the total coverage at each step 
to that at $t = 0$. 
Fig.~\ref{fig:exp1-U} expresses the processes
of the ratio with different $\Theta$. 
The graph shows that the improvement with higher
$\Theta$ becomes larger, whereas that with $\Theta = -5$ dB
was highly limited. 
This is because with low $\Theta$, 
the total coverage was already sufficient 
even in the initial step, and thus 
the objective function did not exhibit much potential for improvement. 
Meanwhile, when $\Theta = 10$ dB, 
the total coverage became approximately two times of that in the initial step. 
This result indicates that our optimization method
is highly efficient for a high SINR regime. 
We also illustrate the 3D deployment of the UAVs in Fig.~\ref{fig:exp-Theta}
in each case. 
The deployment with $\Theta = 10$ dB was similar to that 
with $\Theta = 0$ dB,
whereas the altitudes of the UAVs with $\Theta = -5$ dB 
were much lower than the other cases.


Finally, we demonstrate the impact of the initial positions of the UAVs. 
Fig.~\ref{fig:exp-1-p} shows the 3D deployment of the UAVs at 
$ t = 0 , 100$ 
when their initial positions were gathered at the
center. The figures show that
as the time elapsed, the positions of the UAVs
adjusted to the hotspots and 
the UAVs maintained certain distances among themselves to reduce interference. 
In addition, the altitudes of the UAVs serving the hotspots
became high again, whereas the resulting deployment 
was different from that in the previous scenario.

\subsection{Effect of crowd density estimation}
We next focus on the effect of the crowd density estimation part. 
Unlike the previous scenario, we assumed that 
the intensity function is available only at 
randomly generated GSs (grids where GSs exist). 
The entire intensity was calculated
by the crowd density estimation method
described in Section~\ref{subsec-crowd}.
Since the number of GSs is likely to impact the performance, 
we evaluated it by considering
the GS deployment ratio $r_{\mathrm{GS}} \in [0,1]$. 
This is defined as the ratio of the grids where
GSs exist to the total grids. 
We regarded each grid as the sensing range of a GS\footnote{In our example, 
the area of each grid 
is 100~m${}^2$ (2500 grids).} and assumed that 
the intensity in each GS's grid was obtained correctly. 
To evaluate the impact of $r_{\mathrm{GS}}$, 
we consider the coverage gain,
which was calculated as the difference between
the total coverage at $t = 0$ and $30$. 
We first calculated the coverage gain in the
case with complete information of true intensity, 
(i.e., $r_{\mathrm{GS}} = 1$) as a baseline. 
We then calculated 
the ratio of the coverage gain 
with each $r_{\mathrm{GS}}$ to the baseline. 
Fig.~\ref{fig:exp2} illustrates the result of 
the ratio of coverage gain. 
For each $r_{\mathrm{GS}}$, we generated 10 samples
of the GS placement. 
The figure shows that as expected,
as $r_{\mathrm{GS}}$ increases, the performance of  
of our method also increases. This is because 
the proposed algorithm optimizes the positions
of UAVs by using $\widetilde{\lambda}(\vc{y})$
(see (\ref{eq-J(u)})). 
Thus, if the intensity function was incorrectly estimated, 
the proposed method failed to update
the UAV placement effectively for improving
the objective function. 
However, the simulation results demonstrate that 
even a small $r_{\mathrm{GS}}$ (e.g., $r_{\mathrm{GS}}=0.03$, 
3 GS/km${}^2$) achieved 98\% of the baseline in the urban scenario. 
This indicates that our method 
could improve the total coverage even with a 
limited number of GSs.

\begin{figure}[!t]
\centering
\includegraphics[width=2.2in]{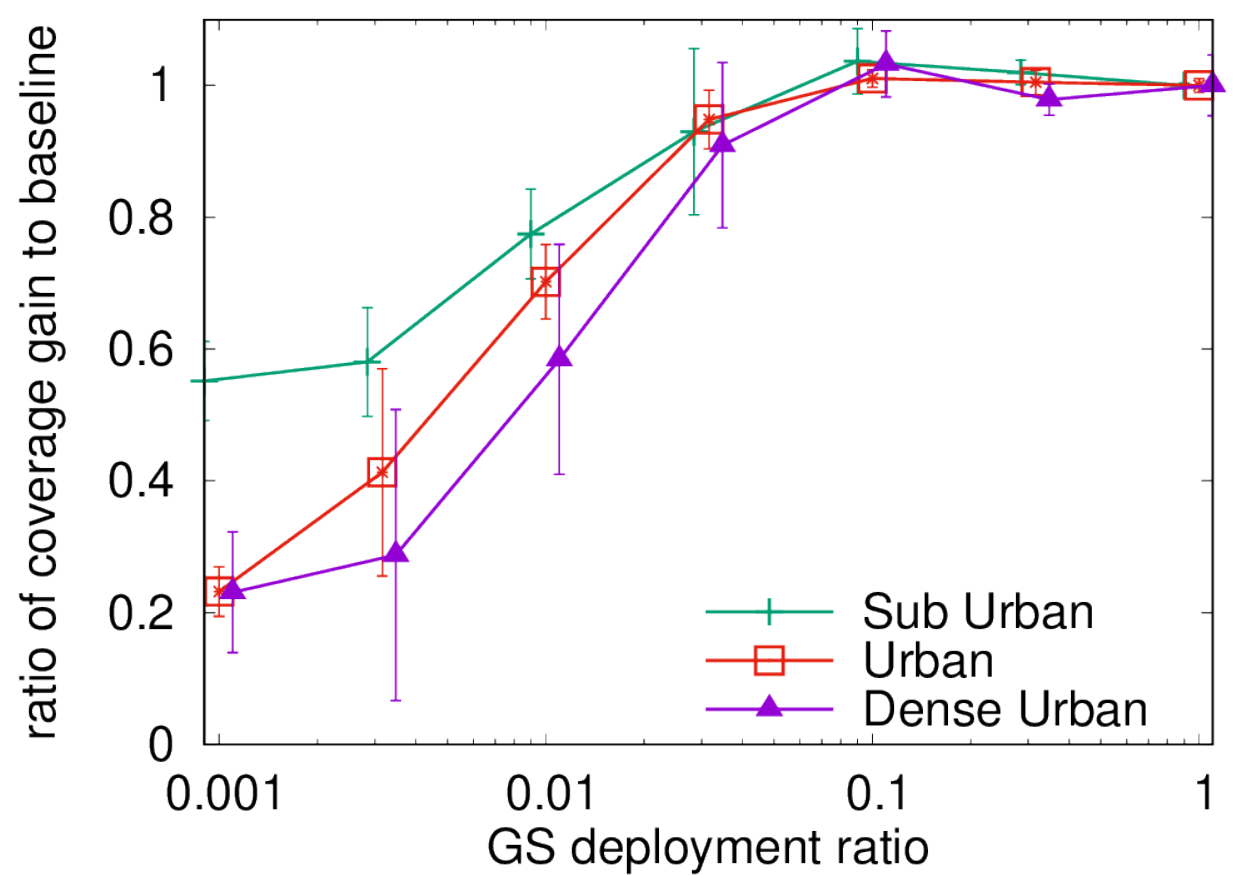}
\caption{
Ratio of coverage gain to baseline versus
GS deployments ratio $r_{GS}$. 
Error bars represent 95~\% confidence interval. 
 }
\label{fig:exp2}
\vspace{-2.5mm}
\end{figure}

\subsection{Dynamic network scenario}\label{subsec-dynamic}

\begin{figure*}[!t]
\centering
\includegraphics[width=6.8in]{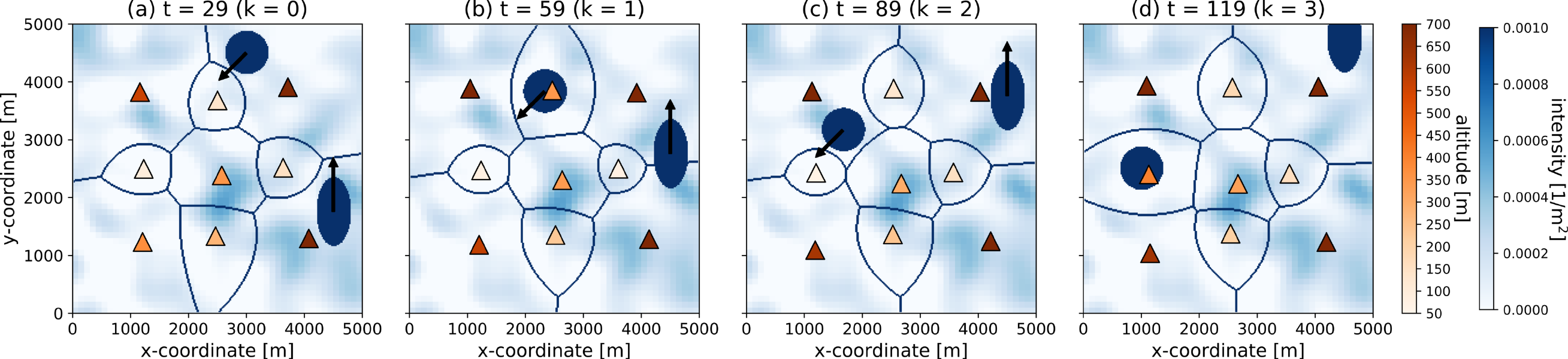}
\caption{
3D UAV deployment process in dynamic network scenario
at (a) $t = 29~(k = 0)$; (b) $t=59~(k = 1)$; (c) $t = 89~(k = 2)$ and 
(d) $t = 119~(k = 3)$. Initial deployment of UAVs was the 
same as that in Fig.~\ref{fig:exp-1-U-Theta0_step}~(a). 
Hotspots are represented as dark ellipses, and arrows represent their moving directions. 
}
\label{fig:exp-3-U}
\vskip -7pt
\end{figure*}

\begin{figure}[!t]
\centering
\includegraphics[width=2.2in]{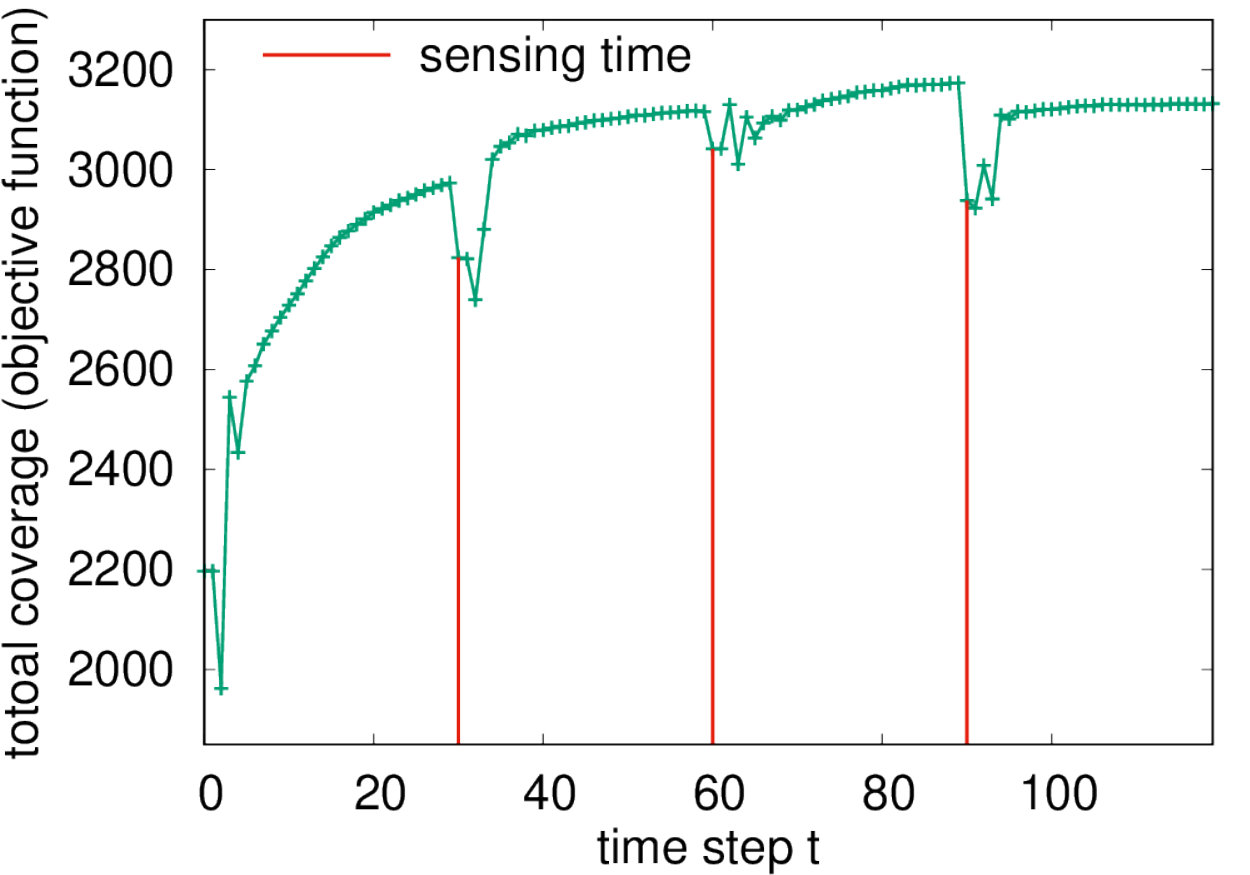}
\caption{Process of total coverage. 
Vertical lines represent sensing times. 
 }
\label{fig:exp3}
\vspace{-2.5mm}
\end{figure}

We next apply our algorithm to a more dynamic network,
in which the positions of hotspots change dynamically. 
Unlike the previous examples, 
we consider multiple sensing periods and 
assume that  hotspots move at each sensing time. 
We first artificially set hotspots on $\Omega_2$
similarly as those in Fig.~\ref{fig:exp-1-U-Theta0_step}. 
We then moved each of them in a predefined direction
at $T_k$ $(k =1,2,3)$ with $T = 30$. 
The initial deployment of UAVs $(T_0 = 0)$ 
was also the same as Fig.~\ref{fig:exp-1-U-Theta0_step}~(a). 
Fig.~\ref{fig:exp-3-U} illustrates
the 3D deployment of UAVs at the end
of each sensing period (e.g., 
(b) shows $t =59$, the end of the first sensing period ($k=1$)). 
%
We also show the process of the total coverage in Fig.~\ref{fig:exp3}. 
As the figures show, 
at each sensing time, the total coverage temporarily 
decreased due to the movement of the hotspots. However, 
as shown in Fig.~\ref{fig:exp-3-U},  
the UAVs adjusted their positions to 
the hotspot locations and
successfully increased the total coverage. 
For example, in Fig.~\ref{fig:exp-3-U} (b) and (d), (i.e., $k = 1,3$), 
the left hotspot moved to a position on the borders 
of UAV-cells. However, the UAVs changed their 3D positions 
and the shapes of the UAV-cells to reduce interference. 
Note that, at $t = 59$ (b),
the right hotspot is on the border of UAV cells. 
This is because our algorithm could not sufficiently optimize
the UAV placement during a sensing period. However, 
the total coverage  was successfully improved as 
shown in Fig.~\ref{fig:exp3}. 
Consequently, our method could respond to hotspots and 
provide coverage to UEs in a dynamic network 
in an on-demand manner.

\section{Conclusion}\label{sec-conc}
In this paper, we proposed a distributed 
UAV-BS 3D deployment method for on-demand
coverage. 
Since the specific positions of ground users
may not be obtained in real networks, 
we proposed the sensing-aided crowd 
density estimation method. 
Furthermore, by adopting total coverage 
as the performance metric, 
we developed a push-sum-based 
distributed UAV 3D deployment algorithm. 
Several simulation results revealed that
our method can improve the overall coverage
of ground users with the aid of a limited number of sensors
and can be applied to a dynamic network. 

In our problem setting, we assumed that the 
UAVs exhibit identical transmission power. Thus, 
the joint optimization of 3D UAV deployment
and energy consumption 
are topics for our future work. 
Furthermore, other physical constraints, 
such as  the capacity of UAVs~\cite{Kala16}, and
obstacle avoidance~\cite{Zhao18}
should be considered while applying 
our method to a real network. 
Moreover, simulations
considering more realistic user density,
such as road constraints, remain for future work.

  \section*{Acknowledgment}

This work was supported by JSPS KAKENHI Grant Numbers JP19K14981 and JP18K13777.



\section*{Appendix A: Proof of Lemma~\ref{lem-A}}\label{appen-proof-of-Lemma-lem-A}
By conditioning on the status of 
the desired channel between $\vc{u}_i$ and $\vc{y}$, 
$C_{i,y}(\vc{u})$ can be rewritten as
\begin{align}
&C_{i,y}(\vc{u}) =
\sum_{q_0 \in \{\rmL, \rmN\}}
\PP(q_0; \theta_{i,y})\PP(\SINR_{i,y} > \Theta \mid q_0).
\label{eq-A_{i,n}-add}
\end{align}
Since small-scale fading gain is distributed according to 
the normalized gamma distribution with the parameter $m_{q_0}$, 
by using (\ref{eq-SINR}), 
we obtain, for $q_0 \in \{\rmL, \rmN\}$
\begin{align}
&\PP(\SINR_{i,y} > \Theta \mid q_0)=
\PP\left(
{
h_{i,y}
}
\left.
> {\Theta (I_{i,y} + \sigma) \over  \ell_{q_0}(d_{i,y})}
\right|
q_0
\right)
\nonumber
\\
&\overset{(a)}{\approx}
1- 
\bbE_{I_{i,y}}\left[
\left(
1- \exp
\left(
- {\eta_{q_0}\Theta (I_{i,y} + \sigma) \over  \ell_{q_0}(d_{i,y})}
\right)
\right)^{m_{q_0}}
\right]
\nonumber
\\
&=
\sum_{k=1}^{m_{q_0}} (-1)^{k+1} {m_{q_0} \choose k}
e^{-k \sigma \gamma_{q_0, i, y}}
\bbE_{I_{i,y}}\left[
e^{
- {\gamma_{q_0, i, y} I_{i,y}}
}
\right]
\nonumber
\\
&\overset{(b)}{=}
\sum_{k=1}^{m_{q_0}} (-1)^{k+1} {m_{q_0} \choose k}
e^{-k \sigma \gamma_{q_0, i, y}}
\calL_{I_{i,y}}(k\gamma_{q_0, i, y}),
\label{eq-SINR-add-1}
\end{align}
where we approximate the tail probability of the normalized
gamma distribution in (a) by using the same method as in
\cite{Bai15, Zhu18}, 
and we use the Laplace transform of $I_{i,y}$ in (b). 
Since the small-scale fading gain and channel condition are i.i.d.
for each UAV--UE channel, 
$\calL_{I_{i,y}}(s)$ can be calculated by
\begin{align}
\calL_{I_{i,y}}(s) &= 
\bbE\left[\exp\left(- s \sum_{j \in \calU \backslash\{i\}}
h_{j,y}\ell_{q_j}(d_{j,y})
\right)\right]
\nonumber
\\
&=
\prod_{j \in \calU \backslash\{i\}}
\bbE\left[
\exp\left( - s 
h_{j,y} \ell_q(\|\vc{u}_j - \vc{y}\|)
\right)
\right]
\nonumber
\\
&=
\prod_{j \in \calU \backslash\{i\}}
\sum_{q_j \in \{\rmL, \rmN\}}\!\!\!
\PP(q_j; \theta_{j,y})
\bbE\left[
e^{-s h_{j,y} \ell_{q_j}(d_{j,y})}
\mid q_j
\right],
\nonumber
\\
&=
\prod_{j \in \calU \backslash\{i\}}
\sum_{q \in \{\LoS, \NLoS\}}\!\!\!
\PP(q; \theta_{j,y})
\left(
1 + {s \ell_{q}(d_{j,y}) \over m_{q}}
\right)^{-m_{q}},
\nonumber
\end{align}
%
where the last equality follows from the Laplace transform 
of a gamma distribution with a parameter $m_q$. 
Finally, combining the above with (\ref{eq-SINR-add-1})
and substituting it into (\ref{eq-A_{i,n}-add}), we obtain (\ref{eq-A_{i,n}}). 

\section*{Appendix B: Proof of Lemma~\ref{lem-div-F}}\label{appen-proof-of-{lem-div-F}}
In this appendix, we provide an outline of the proof due to space limitations. 
First, for each UAV $j \in \calU$,
the boundary of $\calV_j(\vc{u})$ is determined
by only UAV $j$ and its neighbors $\calN_j$
according to the definition of $\calV_j(\vc{u})$. 
Thus, by differentiating $\widetilde{F}(\vc{u})$ with respect to 
$\vc{u}_{j}$, we obtain
\begin{align}
{\rmd \widetilde{F}(\vc{u}) \over \rmd \vc{u}_{j}}
&=
\sum_{i \in \calN_j \cup \{j\}}
{\rmd  \over \rmd \vc{u}_{j}}
\int_{\calV_i(\vc{u})} 
\widetilde{C}_{i,y}(\vc{u}) 
\widetilde{\lambda}(\vc{y})
\rmd \vc{y}
\nonumber
\\
&~~~~+
\sum_{i \notin \calN_j}
\int_{\calV_i(\vc{u})} 
\widetilde{\vc{c}}^j_{i,y}(\vc{u})
\widetilde{\lambda}(\vc{y})
\rmd \vc{y}.
\nonumber
\end{align} 
Therefore, it suffices to prove that 
\begin{align}
&\sum_{i \in \calN_j \cup \{j\}}\!\!
{\rmd  \over \rmd \vc{u}_{j}}
\int_{\calV_i(\vc{u})}
\widetilde{C}_{i,y}(\vc{u}) 
\widetilde{\lambda}(\vc{y})
\rmd \vc{y}
\nonumber
\\
&\qquad
\approx
\sum_{i \in \calN_j \cup \{j\}}
\int_{\calV_i(\vc{u})}
\widetilde{\vc{c}}_{i,y}^j(\vc{u})
\widetilde{\lambda}(\vc{y})
\rmd \vc{y}. 
\label{eq-div-J-add-1}
\end{align}
%
%
%
Note that $\widetilde{C}_{i,y}(\vc{u})\widetilde{\lambda}(\vc{y})$ 
is continuously differentiable at any $\vc{u} \in (\Omega_3)^U$ 
because $\widetilde{C}_{i,y}(\vc{u})$ is continuously differentiable in $(\Omega_3)^U$
and $\widetilde{\lambda}(\vc{y})$ is bounded on $\calV_i(\vc{u})$. 
Furthermore, for any fixed $\vc{u}' \in (\Omega_3)^U$, 
$\widetilde{C}_{i,y}(\vc{u}')\widetilde{\lambda}(\vc{y})$ and 
$\widetilde{\vc{c}}_{i,y}^j(\vc{u}')\widetilde{\lambda}(\vc{y})$
are both integrable on $\Omega_2$
because $\widetilde{C}_{i,y}(\vc{u})$, $\vc{c}_{i,y}^j(\vc{u})$
and $\widetilde{\lambda}(\vc{y})$ are all bounded on $\Omega_2$. 
Furthermore, according to (\ref{eq-S_{i,y}}) and (\ref{eq-def-calC_i}), we can demonstrate that 
each $\calV_i(\vc{u})$ $(i\in\calU)$ can be represented as a set difference of 
certain strictly star-shaped sets. This is because $\Omega_2$ is a convex set and
the average signal power $S_{i,y}(t)$ is a monotonic decreasing 
and convex function of $d_{i,y}(t)$
(see (\ref{eq-S_{i,y}})). 
Therefore, we can apply Proposition~A.1 in \cite{Cort05} (see also Remark~A.2 therein)
to each integral on the left side of (\ref{eq-div-J-add-1}). 
As a result, 
 $\int_{\calV_i(\vc{u})} \widetilde{C}_{i,y}(\vc{u}) \widetilde{\lambda}(\vc{y})\rmd \vc{y}$
($i \in \calN_j\cup\{j\}$)
is continuously differentiable in $(\Omega_3)^U$. 

Moreover, we can approximate $\SINR_{i,y}(t) \approx \SINR_{j,y}(t)$ 
for neighboring UAVs $i,\ j \in \calU$
and a UE at $\vc{y}$ on the border of their UAV cells
$\calV_i$ and $\calV_j$ (see also (\ref{eq-def-calC_i}) and (\ref{eq-SINR})). 
Thus, 
by proceeding in the same manner as 
in the proof of Theorem~2.2 in \cite{Cort05}, 
we can prove that (\ref{eq-div-J-add-1}) holds, which completes the proof. 
\end{document}